\documentclass[letterpaper, 10 pt, journal, twoside]{ieeetran}
\IEEEoverridecommandlockouts
\usepackage{graphics} % for pdf, bitmapped graphics files
\usepackage{epsfig} % for postscript graphics files
\usepackage{supertabular}
\usepackage{graphicx,bbold}
\usepackage{amsfonts}										
\usepackage{amssymb}
\usepackage{mathtools}
\usepackage{times}
\usepackage{epstopdf}
\usepackage{csquotes}
\usepackage{color}
\usepackage{amsmath,amsthm}
\usepackage{amssymb}
\usepackage{float}
\usepackage{pgfplots}
\pgfplotsset{compat=1.5}
\usepackage{cite}
\newtheorem{remark}{Remark}
\newtheorem{theorem}{Theorem}
\newtheorem{assumption}{Assumption}

\newtheorem{objective}{Objective}

\newtheorem{proposition}[theorem]{Proposition}
\newtheorem{lemma}{Lemma}
\usepackage{subcaption}
\usepackage{color}
\usepackage{tabularx}										
\usepackage{longtable}	
\usepackage{booktabs}
\usepackage{siunitx}										

\newcommand{\bR} { {\mathbb R}}
\newcommand{\cB} { {\mathcal B}}

\usepackage{tikz}
\usetikzlibrary{arrows,automata}
\usetikzlibrary{arrows,shapes,backgrounds,calc,positioning,patterns}
\usepackage{balance}
\usetikzlibrary{calc,arrows,shapes,backgrounds,calc,positioning,patterns,decorations.pathmorphing,decorations.markings,mindmap,trees}
\tikzstyle{block} = [draw, rectangle, minimum height=2em, minimum
width=4em] \tikzstyle{sum} = [draw, fill=blue!20, circle, node
distance=1cm] \tikzstyle{input} = [coordinate] \tikzstyle{output} =
[coordinate] \tikzstyle{pinstyle} = [pin edge={to-,thin,black}]
\usepackage[american,cute inductors,smartlabels]{circuitikz}
\usetikzlibrary{arrows,shapes,backgrounds,calc,positioning,patterns}

\usepackage[american,cuteinductors,smartlabels]{circuitikz}

\usetikzlibrary{calc}
\ctikzset{bipoles/thickness=1} \ctikzset{bipoles/length=0.8cm}
\ctikzset{bipoles/diode/height=.375}
\ctikzset{bipoles/diode/width=.3}
\ctikzset{tripoles/thyristor/height=.8}
\ctikzset{tripoles/thyristor/width=1}
\ctikzset{bipoles/vsourceam/height/.initial=.7}
\ctikzset{bipoles/vsourceam/width/.initial=.7} \tikzstyle{every
node}=[font=\small] \tikzstyle{every path}=[line width=0.8pt,line
cap=round,line join=round]
\makeatletter

\begin{document}
\title{Safety during Transient Response in Direct Current Microgrids using Control Barrier Functions}
\author{ K. C. Kosaraju, S. Sivaranjani, and V. Gupta
\thanks{K.C. Kosaraju, V. Gupta are with the Department of Electrical Engineering, University of Notre Dame, Notre Dame, IN 46556, USA (email: \{kkosaraj, vgupta2\}@nd.edu). S. Sivaranjani is with the Department of Electrical and Computer Engineering, Texas A\&M University, College Station, TX 77843, USA (email: sivaranjani@tamu.edu).}
}

\maketitle

\IEEEpeerreviewmaketitle
\thispagestyle{empty}
\pagestyle{empty}
\begin{abstract}
We consider the problem of guaranteeing that the  transient voltages and currents stay within prescribed bounds in Direct Current (DC) microgrids, when the controller does not have access to accurate system dynamics due to the load being unknown and/or time-varying. To achieve this, we propose an optimization based controller design using control barrier functions. We show that the proposed controller has a decentralized structure and is robust with respect to the uncertainty in the precise values of the system parameters, such as the load.
\end{abstract}
\vspace{-0.3cm}
\section{Introduction}
The design and operation of microgrids, which are interconnected clusters of Distributed Generation Units (DGUs), loads and energy storage units interacting with each other through distribution lines, has been widely studied in the literature. This paper focuses on Direct Current (DC) microgrids%, which have found widespread applications are generally considered to be more efficient and reliable than Alternating Current (AC) microgrids 
~\cite{trip2018distributed,sadabadi2020line,cavanagh2017transient,justo2013ac}. The main control objective in DC microgrids is to ensure that the load voltage is stabilized to a desired value \cite{dragivcevic2015dc}. To achieve this objective, several control methods have been proposed, e.g., droop control \cite{zhao2015distributed}, plug-and-play control \cite{tucci2017line}, sliding mode control~\cite{cucuzzella2018sliding}, passivity-based control~\cite{kosaraju2020differentiation}, output regulation \cite{9129742} and input-to-state stability based control~\cite{iovine2018voltage}. 

However, there is limited literature studying the problem of maintaining the voltages and currents in the system within some prescribed bounds during transient operation in DC microgrids. Violating such constraints could lead to deterioration of equipment, ultimately leading to its failure. In this paper, we consider this problem. There are at least two challenges here. First, the control strategy should be decentralized (or at least distributed), for reasons of scalability and robustness. Second, since the knowledge of system parameters (such as the load values) is always imperfect, the controller should not rely on precise values of such parameters being available. As a first step towards solving the more general problem, we design a decentralized control algorithm for voltage and control regulation during transient operation for an islanded DC microgrid with purely resistive lines that does not require accurate information of load values. 

Our controller design relies on casting the problem as one of safety and utilizing Control Barrier Functions (CBFs) to ensure that the system trajectories remain in a desired safe set. CBFs are now a widely accepted tool for designing safety based controllers \cite{XU201554, ames2019control, xiao2019control}. CBFs  guarantee  the existence  of  control  inputs  under  which  a  super-level  set  of a function (typically representing a specification  such as safety) is forward invariant under a given dynamics. However, they do not seem to have been widely explored in the context of DC microgrids. Our main contribution over the existing literature is that we utilize CBFs to design an optimization based controller for DC microgrids that guarantees constraint satisfaction during transient response, while being decentralized and not requiring a precise knowledge of the system parameters. % with the associated challenges of imperfect knowledge of system parameters and need for a decentralized control structure. By explicitly solving these challenges, we are able to utilize CBFs to .

%Our main contribution over the existing literature is that our controller ensures that the transient voltage and current signals stay within specified bounds. We utilize CBFs for this purpose. To solve the challenge that we do not have access to an accurate model of the system dynamics since the load values in a DC microgrid are usually unknown and time-varying, we select an appropriate class $\mathcal{K}$ function in the CBF methodology. The designed controller is robust with respect to uncertainties in the system parameters and can be implemented in a decentralized fashion.

The paper is organized as follows.  In Section~\ref{sec:problem}, we present the model of the DC microgrid and provide the problem formulation.  In Section~\ref{sec:results}, we propose a new control design using control barrier functions that solves the problem. In Section~ \ref{sec:sim}, we corroborate the proposed control design in simulations using a DC microgrid with four buck converters. Finally, in Section \ref{sec:conc}, we conclude and provide some directions for future work. For completeness, a short review of control barrier functions is provided in Appendix A. % In Section~\ref{sec:prelims}, we present relevant definitions and results of control barrier functions.

\paragraph*{Notation} $\mathbb{R}^{n}$ denotes the space of $n$-dimensional real vectors and $\mathbb{R}$ denotes the space of real numbers. $\mathbb{1}$ denotes the vector of all ones with the dimension clear from the context. For a vector $v\in \bR^n$, $v^{T}$ denotes its transpose, $\|v\|_{2}$ denotes its 2-norm, $v_{i}$ denotes its $i$-th element, and $[v]$ denotes a diagonal matrix with $v_i$ as the $i^{th}$ diagonal entry, $\forall i \in \left\{1, \ldots,n\right\}$. For two vectors $u,v\in\mathbb{R}^{n}$, the inequality $u\leq v$ is interpreted element wise. A graph $\mathcal{G}=\left(\mathcal{V}, \mathcal{E}\right)$ is defined by a node set $\mathcal{V}$ and an edge set $\mathcal{E}.$ For a set $\mathcal{S}$, $|\mathcal{S}|$ denotes its cardinality. A function $h:\bR^n \rightarrow \bR$ is  said to be of class $C^k$ if the first $k$ derivatives exist and are continuous. A continuous function $h: [0,a) \rightarrow [0,~\infty)$, $a>0$ is said to belong to class $\mathcal{K}$ function if it is strictly increasing and $h(0) = 0$. It is said to belong to class $\mathcal{K}_{\infty}$ if $a=\infty$ and $h(r)\rightarrow \infty$ as $r\rightarrow \infty$. Given $f(x): \bR^n \rightarrow \bR^n$ and $h(x) \in C^1$, $L_fh(x)$ denotes the Lie derivative of $h(x)$ along the direction of $f(x)$. A function $f:A\rightarrow B$ is Lipschitz if there exists a constant $L$ satisfying $\|f(b)-f(a)\|_2 \leq L\|b-a\|_2$, for all $a,~b \in A$ and class $C^{1}$ if it is continuously differentiable. 
%\newpage

\section{DC Microgrid Model and Problem Formulation}\label{sec:problem}
In this section, we introduce the DC microgrid model and formulate the  problem of guaranteeing bounds on voltage and current during transient operation.
\vspace{-0.3cm}
\subsection{DC microgrid model}
\label{sec:model}
\begin{figure}[htbp]
	\begin{center}
		\ctikzset{bipoles/length=0.7cm}
		\begin{circuitikz}[scale=0.95,transform shape]
			\ctikzset{current/distance=1}
			\draw
			% transformators i and j
			node[] (Ti) at (0,0) {}
			node[] (Tj) at ($(5.4,0)$) {}
			% Buck i
			node[ocirc] (Aibattery) at ([xshift=-4.5cm,yshift=0.9cm]Ti) {}
			node[ocirc] (Bibattery) at ([xshift=-4.5cm,yshift=-0.9cm]Ti) {}
			(Aibattery) to [battery, l_=\small{}$V_{s,i}$,*-*] (Bibattery) {}
			node [rectangle,draw,minimum width=0.4cm,minimum height=2.4cm] (Bucki) at ($0.5*(Aibattery)+0.5*(Bibattery)+(0.83,0)$) {\scriptsize{Switch} $i$}
			(Aibattery) to [short] ([xshift=0.25cm]Aibattery)
			(Bibattery) to [short] ([xshift=0.25cm]Bibattery)
			% filter i
			node[ocirc] (Ai) at ($(Aibattery)+(1.585,0.2)$) {}
			node[ocirc] (Bi) at ($(Bibattery)+(1.585,-0.2)$) {}
			(Ai) to [short] ([xshift=-0.18cm]Ai)
			(Bi) to [short] ([xshift=-0.18cm]Bi)
			(Ai) to [L, l={$L_i$}] ($(Ai)+(1.25,0)$){}
			to [short,i={$I_{i}$}]($(Ai)+(1.35,0)$){}
			to [short, l={}]($(Ti)+(0,1.1)$){}
			(Bi) to [short] ($(Ti)+(0,-1.1)$);
			\begin{scope}[shorten >= 10pt,shorten <= 10pt,]
				\draw[<-] (Ai) -- node[right] {$u_{i}$} (Bi);
			\end{scope};
			\draw
			% PCC-i
			($(Ti)+(-0.7,1.1)$) node[anchor=north]{{$V_{i}$}}
			($(Ti)+(-0.7,1.1)$) node[anchor=south]{}
			($(Ti)+(-0.7,1.1)$) node[ocirc](PCCi){}
			($(Ti)+(0,1.1)$) to [R, l={$G_{i}$}]($(Ti)+(0,-1.1)$)
			%($(Ti)+(1.0,1.1)$) to [short,i>=\small{$I_{li}$}]($(Ti)+(1.0,0.5)$)to [I]($(Ti)+(1.0,-1.1)$)
			($(Ti)+(-1.1,1.1)$) to [C, l_={$C_{i}$}] ($(Ti)+(-1.1,-1.1)$)
			% line
			($(Ti)+(0,1.1)$) to [short] ($(Ti)+(1.1,1.1)$)
			($(Ti)+(1.1,1.1)$) to [short] ($(Ti)+(1.5,1.1)$)
			%($(Ti)+(1.5,1.1)$) to [R, l={$G_{i}$}]($(Ti)+(2.0,1.1)$)
			
			%($(Ti)+(2.6,1.1)$) to [short,i_={$I_{ij}$}] ($(Ti)+(2.75,1.1)$)--($(Ti)+(2.7,1.1)$)
			($(Ti)+(1.5,1.1)$)--($(Ti)+(1.9,1.1)$) to [R, l={$R_{ij}$}] 
			%($(Ti)+(2.5,1.1)$) {}	to [L, l={{$L_{ij}$}}, color=black]($(Tj)+(-1.5,1.1)$){}
			($(Tj)+(-1.5,1.1)$) to [short] ($(Tj)+(-1.5,1.1)$)--($(Tj)+(-1.5,1.1)$)
			($(Ti)+(0,-1.1)$) to [short] ($(Tj)+(-1.5,-1.1)$);
			\draw
			node [rectangle,draw=none,minimum width=6.5cm,minimum height=3.6cm,dashed,fill=blue,opacity=0.1,label=\textbf{DGU $i$},densely dashed, rounded corners] (DGUi) at ($0.6*(Aibattery)+0.5*(Bibattery)+(2.9,0.4)$) {}
			%node [rectangle,draw=none,minimum width=1.75cm,minimum height=2.6cm,fill=red,opacity=0.2,dashed,label=\textbf{ZI Load $i$},densely dashed, rounded corners] (DGUi) at ($0.5*(Aibattery)+0.5*(Bibattery)+(5.15,0)$) {}
			node [rectangle,draw=none,minimum width=1.9cm,minimum height=3.6cm,fill=gray,opacity=0.3,label=\textbf{Line $ij$}, rounded corners] (DGUi) at ($0.5*(Aibattery)+0.5*(Bibattery)+(7.2,0.4)$) {};
		\end{circuitikz}
		\caption{Electrical scheme of DGU $i \in \mathcal{V}$ and line $k \sim \{i, j\} \in \mathcal{E}$, $j \in \mathcal{N}_i$, where $\mathcal{N}_i$ is the set of the DGUs connected to DGU $i$.}
		\label{fig:microgrid}
	\end{center}
\end{figure}
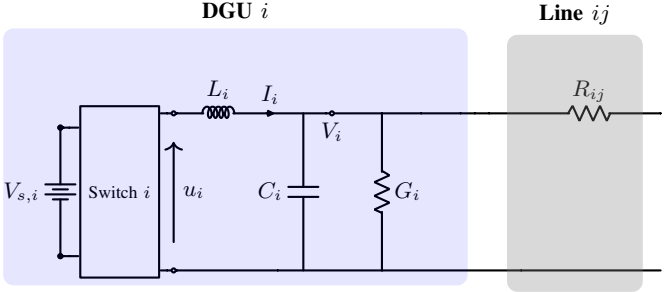
We model a DC microgrid by a connected undirected graph  $\mathcal{G} = \left(\mathcal{V}, \mathcal{E}\right)$, where each node $i\in\mathcal{V}$ represents a distributed generating unit (DGU) containing a voltage source, buck converter to step-down the voltage, and a load, and each edge $e\in\mathcal{E}$ represents a transmission line interconnecting the corresponding DGUs. Let $|\mathcal{V}|=n$ and $|\mathcal{E}|=m.$ For each edge $(u,v)\in\mathcal{E}$, arbitrarily choose one of the ends and assign it the value $+1$; similarly, assign the value $-1$ to the other end. Define the incidence matrix of the graph $\mathcal{G}$ by  $\mathcal{B}\in \bR^{n \times m}$ through the relation:
\vspace{-0.3cm}
    \begin{equation*}
    \cB_{ik}=
    \begin{cases}
    +1 \quad &\text{if $(i,k)\in\mathcal{E}$ and $+1$ was assigned to $i$}\\
    -1 \quad &\text{if $(i,k)\in\mathcal{E}$ and $-1$ was assigned to $i$}\\
    0 \quad &\text{otherwise}.
    \end{cases}
    \end{equation*} 
Note that, by construction,
\vspace{-0.2cm}
\begin{equation}
    \mathcal{B}^\top 1 =0.
\end{equation}
We assume that each transmission line is purely resistive~\cite{lasseter2010certs}. A schematic of each DGU is provided in Figure \ref{fig:microgrid} and the variables used to specify the dynamic evolution of the microgrid are summarized in Table~\ref{tab:symbols}.

	\begin{table}[hbtp]
		\begin{center}
			\caption{A summary of the used variables}\label{tab:symbols}
			\begin{tabular}{clcl}
				\toprule
				& {\bf States and input} & &{\bf Parameters}\\
				\midrule
				$I_{i}$						& Generated current &$ L_{i}$ & Filter  inductance\\
				$V_i$						& Load voltage &$C_{i}$ &Shunt capacitor\\
				$u_i$						& Control input &$G_{i}$ & Load impedance\\ %and its upper, lower bounds.\\
				& &$G_{li},~G_{hi}$ & Bounds for load impedance\\
				$V_{s,i}$& Voltage source &$R_{tk}$ & Line resistance\\
				\bottomrule
			\end{tabular}
		\end{center}
	\end{table}
% 	\vspace{-0.3cm}
The state of the DC microgrid is specified via the voltage at every node and current through every edge. At each node $i$, denote the current through the inductor $L_{i}$ by $I_{i}$, and denote the voltage across the load impedance $G_{i}$ by $V_{i}$. Stack these variables for all the nodes to form the current and voltage vectors $I,~V:\bR_+\rightarrow \bR^n$.  The control input at the $i$-th node is given by the duty-ratio of the buck converter of the $i$-th DGU, which is denoted by $u_{i}\in(0,1].$ %Stack the control input for all the nodes to define the control input for the entire system as  $u :\bR_+\rightarrow (0,1]^n.$ The external input for the $i$-th node is given by the source voltage $V_{s,i}$. Stack the source voltages into the external input vector $V_{s}$.
$V_{s,i}\in \bR_{+}$ denotes the voltage source connected to the $i$-th node.
Stack the control inputs and the voltage sources for all the nodes to define the control input and the voltage source for the entire system as  $u :\bR_+\rightarrow (0,1]^n$ and $V_{s} \in \bR^n_+$, respectively.
Finally, define positive definite and diagonal matrices $L,~C,~G \in \mathbb{R}^{n\times n}$ and $R_t \in \mathbb{R}^{m\times m}$ obtained by stacking and diagonalizing the resulting vector of DGU inductances, capacitances, load impedances and line resistances.

With these quantities defined, the dynamics governing the DC microgrid are given by the relations
\begin{align}\label{model:DC_microgid}
\begin{split}
        -L\dot{I}&=  V -V_{s}u\\
    C\dot{V}&= I - GV - \mathcal{B} R^{-1}_t\mathcal{B}^\top V.
\end{split}
\end{align}
For future reference, denote
\begin{align}\label{denote:Gp}
    G_p\triangleq G + \mathcal{B} R^{-1}_t\mathcal{B}^\top.
\end{align}
We make the following assumptions in the paper.
\begin{assumption}\label{ass:available_information} 
The current $I_i$ and the voltage $V_i$ are available by direct measurements at each DGU $i\in \mathcal V$.
\end{assumption}
\begin{assumption}\label{ass:load_uncertainity} 
While the exact value of the load $G$ is unknown,  upper and lower bounds are known as
\begin{align}
G_l\leq G\leq G_h,
\end{align}
where $G_l,~G_h >0$ and the inequalities are interpreted elementwise.
\end{assumption}
\vspace{-0.5cm}
\subsection{Problem formulation}
The control objective in the microgrid defined by the equations \eqref{model:DC_microgid} is to design the control input $u$ such that the voltage $V$ across the load $G$ is stabilized to a desired value. In our formulation, we have assumed that the exact value of the load $G$ is unknown. Thus, the traditional objective of ensuring a specified value for the voltage may be too stringent. Instead, we assume that each load has a safe operating region in terms of a permitted lower bound $v_l\in \bR$ and upper bound $v_h\in\bR$ for the voltage across the load $G_{i}$. Consequently, we define the first control objective as follows:
\vspace{-0.1cm}
\begin{objective}[Safe voltage regulation] \label{obj:Safe_voltage_regulation}
The voltage across the load $G$ must satisfy
\begin{align}\label{eq_obj:Safe_voltage_regulation}
    %v_l 1 \leq V(t) \leq v_h 1.
    v_l\mathbb{1} \leq V(t) \leq v_h\mathbb{1},\qquad\qquad t\geq 0.
\end{align}
\end{objective}
The second objective is to prevent the over or under drawing the current from the source. Thus, we define the second control objective as follows:
\vspace{-0.1cm}
\begin{objective}[Safe current regulation] \label{obj:Safe_current_regulation} The current through the load must satisfy the bounds
\begin{align}\label{obj:Safe_gen_currents}
     I_l\mathbb{1} \leq I(t) \leq I_h\mathbb{1},\qquad\qquad t\geq 0,
 \end{align}
where $I_l, ~I_h \in \bR$ are lower and upper bounds for the allowable values of the current.
\end{objective}
\vspace{-0.1cm}
The problem we consider can now be stated as follows.

\noindent{\bf Problem Statement:} Given the DC microgrid as described in \eqref{model:DC_microgid}, we would like to design decentralized control inputs $u^i$ at each node $i\in \mathcal{V}$ to achieve Objectives~\ref{obj:Safe_voltage_regulation} and~\ref{obj:Safe_current_regulation}, with each $u^{i}$ being computed using only local information about the state variables $I^i,~V^i$ at node $i$.

Since the desired voltage across the load should be less than the supply voltage $V_s$, a trivial bound for $v_{h}$ in~(\ref{eq_obj:Safe_voltage_regulation}) is given by the source vector $V_{s}$. To see this more formally, note that for a given constant input $\overline{u}\in \bR^n$, the corresponding steady state solution $(\overline I, \overline V)$ of \eqref{model:DC_microgid}  satisfies
\begin{align}\label{equilibria}
%    \begin{split}
    - \overline V +V_s\overline u&=0,\\
    \label{equilibria2}
   \overline I - G\overline V - \mathcal{B} R^{-1}_t\mathcal{B}^\top \overline V &=0.
%    \end{split}
\end{align}
Thus, the set of all feasible forced equilibria are given by the tuples $(\overline I,\overline V,\overline u)$ that  satisfy~\eqref{equilibria} and~\ref{equilibria2}. Since $0<\overline{u} \leq 1$, the equation \eqref{equilibria} implies that $\overline{V} \leq V_s$. We formalize this through the following assumption needed for the feasibility of the problem.
\vspace{-0.2cm}
\begin{assumption}\label{ass:buck_converter_critera}
We assume that
% \begin{align}
%     v_l\leq v_h < \min(V_s),
% \end{align}
% where $\min(V_{s})$ denotes the minimum element of the vector $V_{s}$.
\begin{align}
    v_l\leq v_h < V_s.
\end{align}
\end{assumption}
Finally, we assume that the problem is feasible.
\vspace{-0.2cm}
\begin{assumption}
\label{ass:feasibility}
The problem stated above is feasible  in the sense that there exists at least one control sequence achieving Objectives~\ref{obj:Safe_voltage_regulation} and~\ref{obj:Safe_current_regulation}. In particular, the initial conditions for the voltages $V$ and the currents $I$ are in the sets defined in~(\ref{eq_obj:Safe_voltage_regulation}) and~(\ref{obj:Safe_gen_currents}).
\end{assumption}
\vspace{-0.3cm}
\section{Proposed Solution}\label{sec:results}
The problem formulated above poses two difficulties. The first is that of guaranteeing the safety objectives~\ref{obj:Safe_voltage_regulation} and~\ref{obj:Safe_current_regulation}. The second is the limited knowledge of the load $G$. For pedagogical ease, we tackle these issues sequentially and present our solution to the problem formulated above in three steps:
\begin{itemize}
    \item We first show how  Objective~\ref{obj:Safe_voltage_regulation} can be guaranteed when the load $G$ is known.
    \item We then extend the solution to the case when only a lower and an upper bound on the load $G$ are known. 
    \item Finally, we extend the solution to also include Objective~\ref{obj:Safe_current_regulation}.
\end{itemize}
The proofs of all the results are provided in the Appendices.
\vspace{-0.3cm}
\subsection{Guaranteeing Objective~\ref{obj:Safe_voltage_regulation} with known load}
We begin by designing a controller that guarantees objective~\ref{obj:Safe_voltage_regulation} when the load $G$ is known. We utilize Control Barrier Functions (CBFs) for the purpose. A brief introduction to CBFs is provided in Appendix A for the interested reader. % Note that Theorem~\ref{thm:ZCBF} requires the relative degree of the system to be one. In our case, the relative degree of the system is two and thus we . 

Specifically, for all nodes $i \in \mathcal{V}$, we propose the  controller %. Define the vectors 
%\begin{align*}
%    \underline{I}&=\\
%    \overline{I}&=.
%\end{align*}
%The controller is then given 
obtained by solving the following optimization problem:
\begin{align}\label{QP:1}
u_{i,1}^{\mathrm{opt}}= & \arg\min_{a_i \in \bR}\|a_i\|^2\\
\nonumber\text{s.t.}\hspace{2mm} &   a_iV_{s,i} - V_i +\eta_{l,i} (I_i - Gv_{l}\mathbb{1}) \geq 0 \\
\nonumber& -a_iV_{s,i} + V_i  -  \eta_{h,i} (I_i - Gv_{h}\mathbb{1}) \geq 0\\
\nonumber& 0\leq a_i\leq 1,
\end{align}
where $0\leq \eta_{l,i},~\eta_{h,i} \leq 0$ are tuning parameters. 

We now show that this controller stabilizes the system to the safe set \eqref{obj:Safe_voltage_regulation}. 
\begin{theorem}\label{prop:CBF_known_G}
Consider the problem formulated in Section~\ref{sec:problem} with the load $G$ being known. The controller $u_{i,1}^{\mathrm{opt}}$ designed in \eqref{QP:1} ensures that the system \eqref{model:DC_microgid} satisfies Objective \ref{obj:Safe_voltage_regulation}. 
\end{theorem}
\begin{proof}
See Appendix B.
\end{proof}
% \vspace{-0.1cm}
Note that the last constraint in the QP \eqref{QP:1} models the physical limits of the control input. 
% \vspace{-0.1cm}
\begin{remark}[Decentralized controller]
In the proposed optimization based controller, computing $u_{i,1}^{\mathrm{opt}}$ requires only local information of the states $(V_i, I_i)$. Thus, the controller has a decentralized structure as required. 
\end{remark}
%\begin{remark}[Robustness to DC microgrid parameters]
The proposed controller is independent of the exact values of the parameters $L$ and $C$. Thus, even if we have limited knowledge of these parameters, the proposed controller provides some robustness with respect to variation in the values of these parameters. %to our knowledge of the internal parameters of the microgrid.
%As seen in the proof of the Proposition \ref{prop:CBF_known_G}, CBF approach requires the knowledge of systems dynamics. Consequently, the resulting control may lack robustness with respect to the uncertainty in the system parameters. In Proposition \ref{prop:CBF_known_G}, we overcome this by a sensible choice of class $\mathcal{K}$ functions in \eqref{CBF:Known_G:a} and \eqref{CBF:Known_G:e}. This particular choice allowed the controller to be . However, we still need to compute the load accurately, which is usually not available. Consequently, in the next subsection we present an approach that overcomes this issue.
%\end{remark}

\subsection{Guaranteeing Objective~\ref{obj:Safe_voltage_regulation} with unknown load}
While the decentralized controller presented in Proposition \ref{prop:CBF_known_G} ensures that the voltage $V(t)$ lies with in the prescribed safe set, it requires an accurate knowledge of the load $G$. In practice, the exact value of the load is unknown and usually tends to change slowly with time. We now extend the controller presented above to not require this knowledge and hence be robust with respect to the precise value of the load $G$.

Specifically, we consider the controller obtained by solving the following optimization problem:
\begin{align}\label{QP:2}
u_{i,2}^{\mathrm{opt}}= & \arg\min_{a_i \in \bR}\|a_i\|^2\\
\nonumber\text{s.t.}\hspace{2mm} &   a_iV_{s,i} - V_i +\eta_{l,i} (I_i - G_{l}v_{l}\mathbb{1}) \geq 0 \\
\nonumber& -a_iV_{s,i} + V_i  -  \eta_{h,i} (I_i - G_{h}v_{h}\mathbb{1}) \geq 0\\
\nonumber& 0\leq a_i\leq 1,
\end{align}
where $0\leq \eta_{l,i},~\eta_{h,i} \leq 0$ are tuning parameters. We can then state the following result.
\begin{theorem}\label{prop:CBF_unknown_G}
Consider the problem formulated in Section~\ref{sec:problem} with only the bounds $G_{l}$ and $G_{h}$ on the load $G$ being known. The controller $u_{i,2}^{\mathrm{opt}}$ calculated as proposed in \eqref{QP:2} ensures that the system \eqref{model:DC_microgid} satisfies Objective \ref{obj:Safe_voltage_regulation}. 
\end{theorem}
\begin{proof}
See Appendix C.
\end{proof}

\vspace{-0.5cm}
\subsection{Guaranteeing both Objectives~\ref{obj:Safe_voltage_regulation} and~\ref{obj:Safe_current_regulation}}

To satisfy Objective~\ref{obj:Safe_current_regulation}, we can proceed in a similar fashion as above. However, we note that the voltages and currents in each DGU are not independent, and hence, their constraints need to be studied jointly. To this end, we note that if we define
\begin{align}\label{eq:feasible_currrents}
\begin{split}
    \tilde{I}_{l}&\triangleq \max{\left\{v_lG_{l}\mathbb{1}, I_{l}\right\}} \\
    \tilde{I}_{h}&\triangleq \min{\left\{v_hG_{h}\mathbb{1}, I_{h}\right\}},
    \end{split}
\end{align}
where the operations $\max$ and $\min$ are defined elementwise, then satisfying the conditions $\tilde{I}_{l}\geq 0$ and $\tilde{I}_{h}\geq 0$ will ensure that both Objectives~\ref{obj:Safe_voltage_regulation} and~\ref{obj:Safe_current_regulation} are met. 

We consider the controller obtained by solving the following optimization problem:
\begin{align}\label{QP:3}
u_{i,3}^{\mathrm{opt}}= & \arg\min_{a_i \in \bR}\|a_i\|^2\\
\nonumber\text{s.t.}\hspace{2mm} &   a_iV_{s,i} - V_i +\eta_{l,i} (I_i - \tilde{I}_{l,i}\mathbb{1}) \geq 0 \\
\nonumber& -a_iV_{s,i} + V_i  -  \eta_{h,i} (I_i - \tilde{I}_{h,i}) \geq 0\\
\nonumber& 0\leq a_i\leq 1,
\end{align}
where $0\leq \eta_{l,i},~\eta_{h,i} \leq 0$ are tuning parameters and state the following result.
\begin{theorem}\label{prop:CBF_gen_cur}
Consider the problem formulated in Section~\ref{sec:problem}. The controller $u_{i,3}^{\mathrm{opt}}$ calculated as proposed in \eqref{QP:3} ensures that the system \eqref{model:DC_microgid} satisfies both Objectives  \ref{obj:Safe_voltage_regulation} and~\ref{obj:Safe_current_regulation}. 
\end{theorem}
\begin{proof}
The proof is analogous to that of Theorems \ref{prop:CBF_known_G} and \ref{prop:CBF_unknown_G} and is omitted.
\end{proof}
\begin{remark}
Although for notational and pedagogical ease, we have assumed in the above development that the lower bound $I_{l}$ and the upper bound $I_{h}$ for all the DGUs are the same, all the arguments above can be extended to consider heterogeneous bounds. In fact, in the simulation study below, we demonstrate the heterogeneous case.
\end{remark}

\vspace{-0.3cm}
\section{Simulation Results}\label{sec:sim}
\vspace{-0.4cm}
\begin{figure}[tbph]
\centering
\begin{tikzpicture}[scale=0.8,transform shape,->,>=stealth',shorten >=1pt,auto,node distance=3cm,
                    semithick]
  \tikzstyle{every state}=[circle,thick,fill=white,draw=black,text=black]

  \node[state] (A)                    {\num{1}};
  \node[state]         (B) [above right of=A] {\num{2}};
  \node[state]         (D) [below right of=A] {\num{4}};
  \node[state]         (C) [below right of=B] {\num{3}};

  \path (A) edge   [below] node {\hspace{7mm}$I_{l1}$} (B)
  		edge 	     node {$I_{l4}$} (D)
           (B) edge      [below]        node {\hspace{-7mm}$I_{l2}$} (C)
           (C) edge         [above left]     node {$I_{l3}$} (D);

\end{tikzpicture}
\caption{DC microgrid considered with four buck converters. }\label{fig:microgrid_example}
\end{figure}
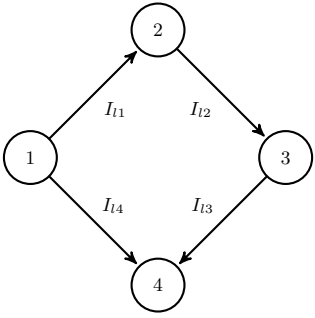

\begin{figure}[t]
	\centering
	\includegraphics[trim=0cm 2cm 2cm 2cm, clip=true, width=\columnwidth]{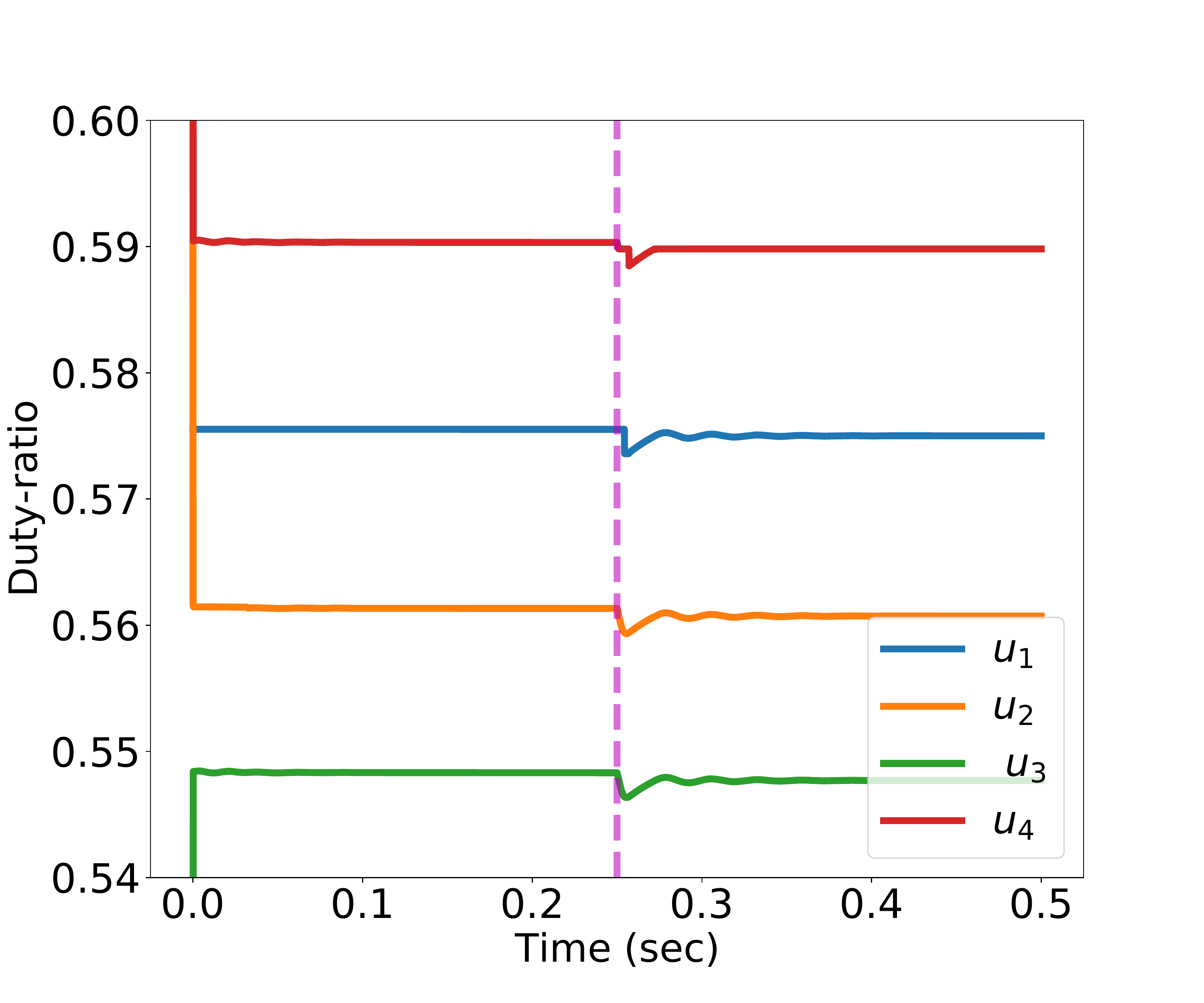}
	\caption{Control input generated by the proposed controller.}
	%\label{fig:sen_a}
	%\vspace{.2cm}
	\label{fig:input}
\end{figure}
In this section, we illustrate the proposed controller through a simulation study. We consider a Kron reduced DC microgrid consisting of four DGUs connected as shown in Figure \ref{fig:microgrid_example}. The parameters of the DGU and distribution lines are reported in Table \ref{tab:parameters1} and Table \ref{tab:parameters2}, respectively. These parameters are similar to those used in \cite{trip2018distributed, kosaraju2020differentiation} and \cite{TRIP2018242} for simulation and experimentation respectively. We assume that the desired nominal voltage for each DGU is \num{230.0} \si{\volt}, with the safe region between \num{229.0} \si{\volt} and \num{231.0} \si{\volt}. %To test the effectiveness of the controller we choose a very narrow safe operating region, i.e., desired $\pm$\num{1} (\si{\volt}). However, in practice these can be a wider than the one considered here. 
Moreover, the precise value of the load parameter $G$ is unknown and only the bounds $G_l =0.95 \times G$ and $G_h = 1.05 \times G$ are known. Finally, for optimal working of the voltage source, we assume that the generating currents are bounded as shown in in Table \ref{tab:parameters1}.
\begin{figure}[t]
	\centering
	\includegraphics[trim=2cm 2cm 1cm 4cm, clip=true, width=1\columnwidth]{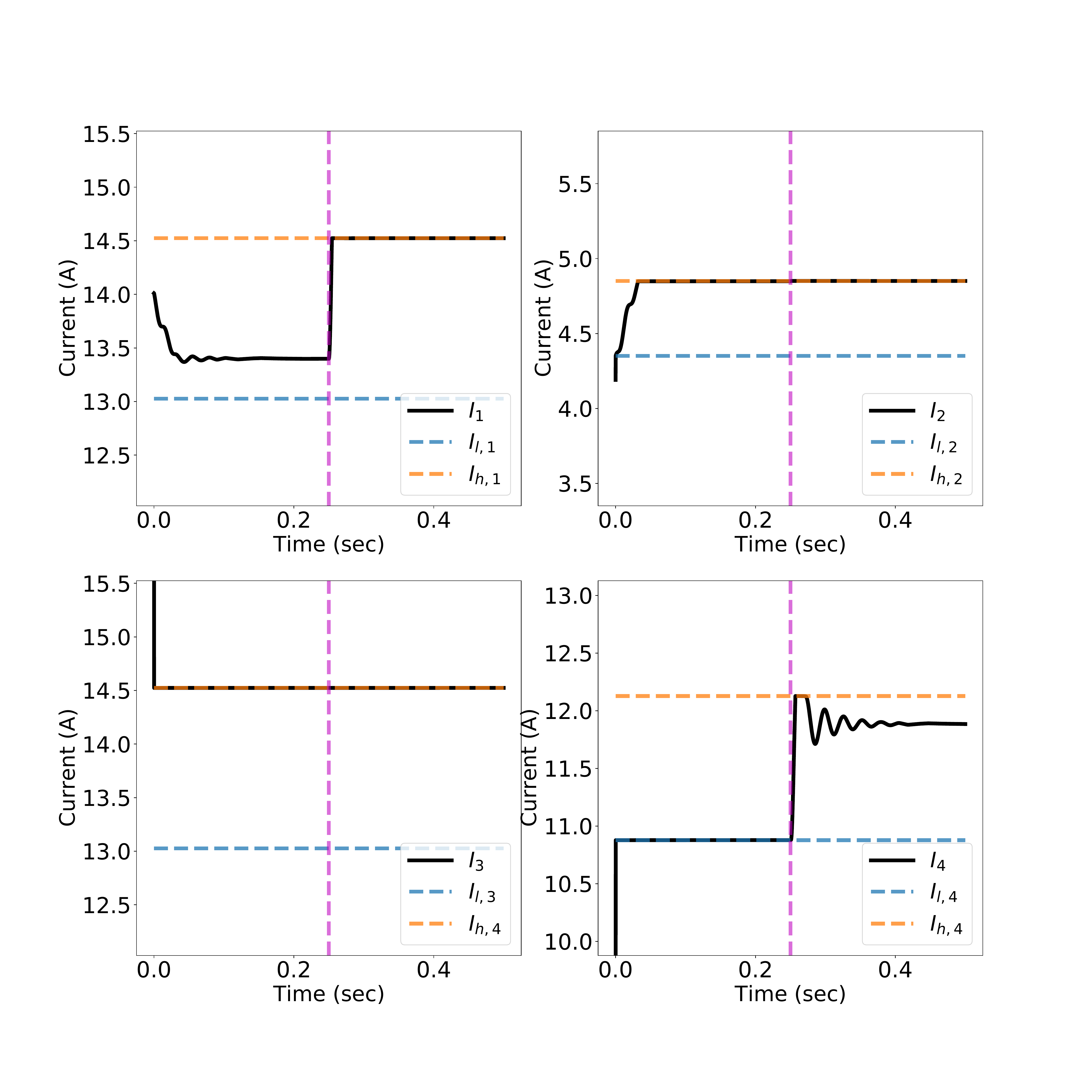}
	\caption{Currents from the voltage source in each DGU.}
	%\label{fig:sen_a}
% 	\vspace{.2cm}
	\label{fig:current}
	\vspace{-0.5cm}
\end{figure}

%The values of the parameters for the controller design used are reported in Table \ref{tab:cont_params}. 
At time $t= 0$, we start the system with a feasible voltage value. However, to test if our method can handle slight violations of the feasibility Assumption~\ref{ass:feasibility}, the generating currents are initialized outside the prescribed safe region. Since the proposed controller assumes initialization within the feasible set, we use a numerical heuristic to bring the desired values quickly within the feasible region. Specifically, we modify the optimization problem \eqref{QP:3} to
\begin{align}\label{QP:load_uncertainity_sol_prac}
u_i^{\mathrm{opt}} = & \arg\min_{a_i, \epsilon_{l,i}, \epsilon_{h,i} \in \bR}a_i^2 + P_{l,i}\epsilon_{l,i}^2 + P_{h,i}\epsilon_{h,i}^2\\
\nonumber\text{s.t.}\hspace{2mm} &   a_iV_{s,i} - V_i +\eta_{l,i} (I_i - \tilde{I}_{l,i})+ \epsilon_{l,i} \geq 0 \\
\nonumber& -a_iV_{s,i} + V_i  -  \eta_{h,i} (I_i - \tilde{I}_{h,i}) + \epsilon_{h,i}\geq 0\\
\nonumber& 0\leq a_i\leq 1,
\end{align}
for large values of constants $P_{l,i}$ and $P_{h,i}$. Once the currents and voltages enter the feasible set, we switch to the controller \eqref{QP:3}.
The simulation is conducted for \num{0.5} \si{\sec}. The tuning parameters of the controller are given in Table \ref{tab:cont_params}. After \num{0.25} \si{\sec} the load at all the nodes is increased by $5\%$. The voltage and current signals of the resulting simulation are plotted in Figure \ref{fig:voltage} and Figure \ref{fig:current}, respectively. The control input generated by the optimization problem is plotted in Figure \ref{fig:input}. We observe that the controller ensures satisfaction of both Objectives~\ref{obj:Safe_voltage_regulation} and~\ref{obj:Safe_current_regulation} (modulo the infeasbility of the specified initial condition) despite the abrupt load change. %voltage signals are well with-in the safety constraints. However in Figure \ref{fig:current}, we see that the generating currents at time \num{0} (\si{\sec}) of DGU 2 and 3 are outside the safety bounds. The controller quickly brings them inside the safe region, while maintaining the voltage within the safe region. To  compensate the abrupt change in load by $5\%$, at time time \num{0.25} (\si{\sec}), the generating currents tries to go beyond the prescribed safe region. However, we see that the controller does a great job in constraining the currents with in the safe region, again without compromising on the safety of load voltages.

%Furthermore, $\epsilon_{l,i}(t),~\epsilon_{h,i}(t)$ are all zeros, $\forall t\in (0,0.3]$.  This is justified from the fact the initial conditions of the current were outside the safe operating region.  This implies, for time $\forall t\in (0,0.3]$, the ZCBF constraints were never violated. This is also reflected in Figure \ref{fig:voltage}, where the voltage signals are always with in the prescribed bounds. Hence the controller not only guarantee that the transient signals are with in the safety set, but also is robust to uncertainty it load parameters.

% \begin{figure}[t]
% 	\centering
% 	\includegraphics[trim=4cm 10cm 4cm 10cm, clip=true, width=\columnwidth]{input_a.pdf}
% 	\caption{}
% 	%\label{fig:sen_a}
% 	\vspace{.2cm}
% 	\label{fig:input}
% \end{figure}
\begin{table}[t]
	\caption{Microgrid Parameters}
	\centering
	{\begin{tabular}{lc | cccc}			
			DGU								&	&1		&2	&3	&4\\			
			\hline
			$L_{i}$	&(\si{\milli\henry})	&\num{1.8}		&\num{2.0}		&\num{3.0}		&\num{2.2}\\
			$C_{i}$	&(\si{\milli\farad})	&\num{2.2}		&\num{1.9}		&\num{2.5}		&\num{1.7}\\
			%{$w_i$	}&{--}	&{$0.4^{-1}$}		&{$0.2^{-1}$}		&{$0.15^{-1}$}		&{$0.25^{-1}$}\\
			$V_i^{\star}$ &(\si{\volt})		&\num{230.0}		&\num{230.0}		&\num{230.0}		&\num{230.0}\\
			$V_{l,i}$ &(\si{\volt})		&\num{229.0}		&\num{229.0}		&\num{229.0}		&\num{229.0}\\
			$V_{h,i}$ &(\si{\volt})		&\num{231.0}		&\num{231.0}		&\num{231.0}		&\num{231.0}\\
			$G_{i}$ &(\si{\ohm})	&\num{1/16.7}		&\num{1/50.0}		&\num{1/16.7}		&\num{1/20.0}\\
			$I_{l,i}$ &(\si{\ampere})	&\num{13.0}		&\num{4.4}		&\num{13.0}		&\num{11.0}\\
			$I_{h,i}$ &(\si{\ampere})	&\num{14.5}		&\num{4.9}		&\num{14.5}		&\num{12.1}
	\end{tabular}}
	\label{tab:parameters1}
	%\end{table}
	\vspace{0.2cm}
	%\begin{table}[t]
	\caption{Line Parameters}
	\centering
	{\begin{tabular}{lc | cccc}			
			Line								&	&1&2	&3		&4\\					
			\hline
			% & & & & &\\
			$R_{k}$	&(\si{\milli\ohm})	&\num{70}		&\num{50}		&\num{80}		&\num{60}
			%\\
			%$L_{k}$	&(\si{\micro\henry})	&\num{2.1}		&\num{2.3}		&\num{2.0}		&\num{1.8}\\
	\end{tabular}}
	\label{tab:parameters2}
\end{table}
\begin{table}[htbp]
\caption{Controller Parameters}
\centering
	{\begin{tabular}{l | cccc}			
			node									&1&2	&3		&4\\					
			\hline
			% & & & & &\\
			$\eta_{l,i}$		&\num{0.5}		&\num{0.4}		&\num{0.5}		&\num{0.3}\\
			$\eta_{h,i}$		&\num{0.4}		&\num{0.3}		&\num{0.5}		&\num{0.4}\\
			$P_{l,i},~P_{h,i}$		&$10^{23}$	&$10^{23}$	&$10^{23}$		&$10^{23}$\\
			%$L_{k}$	&(\si{\micro\henry})	&\num{2.1}		&\num{2.3}		&\num{2.0}		&\num{1.8}\\
	\end{tabular}}
	\label{tab:cont_params}
\end{table}
\begin{figure}[t]
	\centering
	\includegraphics[trim=1cm 6cm 1cm 8.5cm, clip=true, width=1.1\columnwidth]{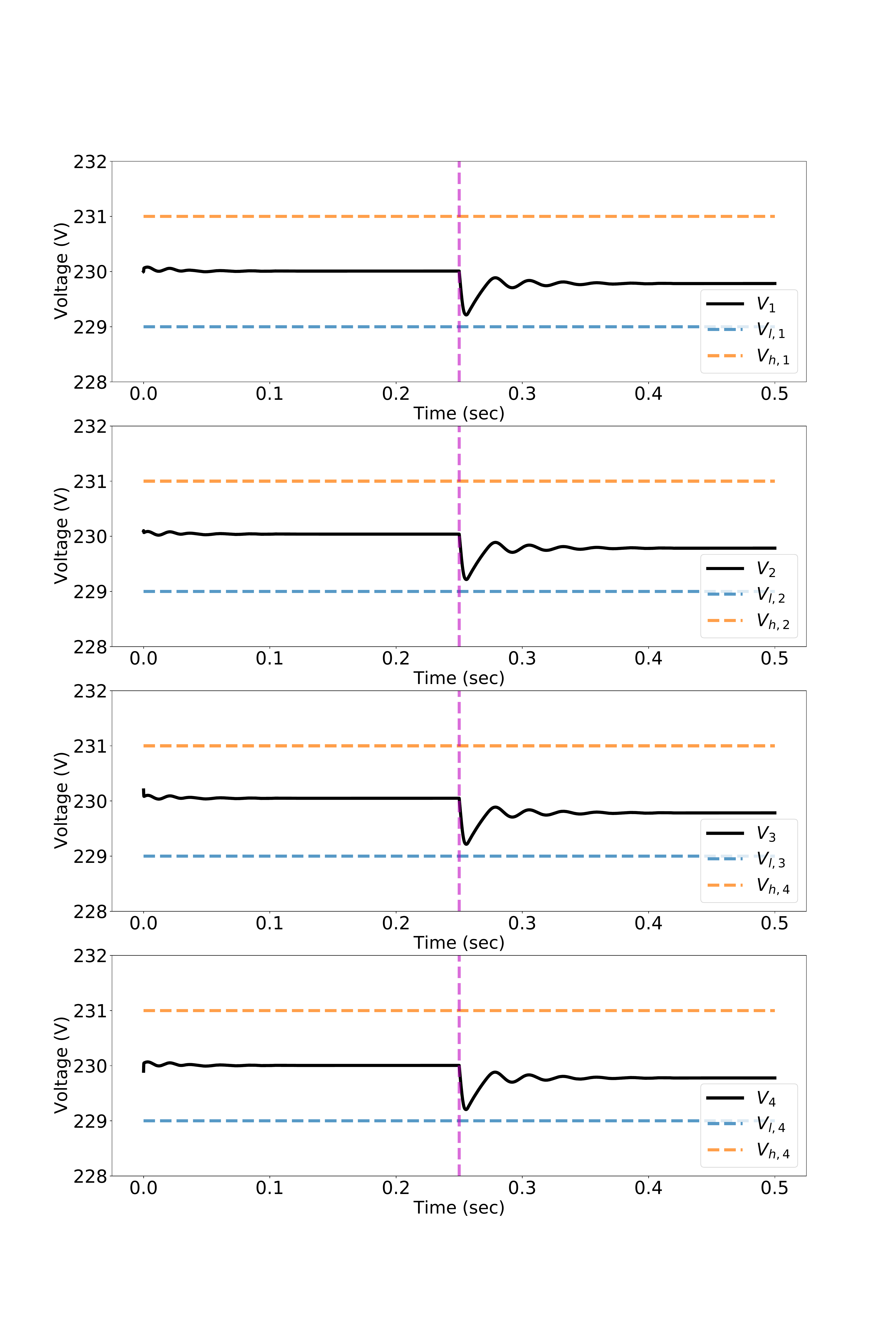}
	\caption{Voltage across the load of each DGU, considering a load variation of $5\%$ at time $t=0.25$ seconds.}
	%\label{fig:sen_a}
	\vspace{.2cm}
	\label{fig:voltage}
\end{figure}
\vspace{-0.15cm}
\section{Conclusions and future work}\label{sec:conc}
% \vspace{-0.1cm}
We presented a new decentralized optimization based controller for DC microgrids. We used Control Barrier Functions to ensure that the voltage and current signals lie within the permitted safety bounds despite the load being unknown and potentially time-varying.  %uncertainty in the load.
Future research directions include extending this technique for ZIP loads and towards achieve safe current sharing in DC microgrids.
\vspace{-0.25cm}
\bibliographystyle{IEEEtran}
\bibliography{reference}
%\begin{appendix}
\section{Appendix A: Control Barrier Functions}
%\section{Preliminaries}\label{sec:prelims}
% Barrier functions are in general used in optimization for enforcing inequality constraints. Such barrier function techniques has been extended to control theory, commonly called as control barrier functions (CBF's). In this brief, we use CBF's to propose a novel safe controller for Direct Current microgrids. Consequently, we introduce  relevant definitions and useful results.
%The following definition of control barrier function is taken from \cite{XU201554}. 
% \begin{definition}[Class $\mathcal{K}$ function]
% A continuous function $h: [0,a) \rightarrow [0,~\infty)$, $a>0$ is said to belong to class $\mathcal{K}$ function if it is strictly increasing and $\alpha(0) = 0$. It is said to belong to class $\mathcal{K}_{\infty}$ if $a=\infty$ and $\alpha(r)\rightarrow \infty$ as $r\rightarrow \infty$.
% \end{definition}
% \begin{definition}[Class $\mathcal{K}_{\infty}$ function]
% A continuous function $h: [0,a) \rightarrow [0,~\infty)$, $a>0$ is said to belong to class $\mathcal{K}_{\infty}$ function if it is a class $\mathcal{K}_{\infty}$ function, and $h(x)\rightarrow \infty$ as $a\rightarrow \infty$.
% \end{definition}
Consider the affine nonlinear system
\begin{align}\label{mode:Nonlinear_system}
    \dot{x}&=f(x)+g(x)u
\end{align}
with $f$ and $g$ are locally Lipschitz continuous, $x\in \bR^n$ and $u\in U\subset \bR^m$.
% \begin{definition}[Forward invariance]
%  A set $\Omega\subset \bR^n$ is forward invariant for system \eqref{mode:Nonlinear_system} if its solutions starting at all $x(t_0)\in \Omega$ satisfy $x(t)\in \Omega$ for $\forall t\geq t_0$. 
% \end{definition}
%\begin{definition}[Zeroing Control Barrier Functions]
%\label{def:zcbf}
 Let $\Omega$ be a zero super level-set of a continuously differentiable function $h:\bR^n\rightarrow \bR$, i.e.,
\begin{align}\label{Levelset:Omega}
    \Omega &=\left\{x \in \bR^n~|~h(x)\geq 0\right\}.
\end{align}
The function $h(x)$ is called a zeroing control barrier function (ZCBF),
if there exists a locally Lipschitz extended class $\mathcal{K}$ function $\alpha$ such that
\begin{align}\label{Barrier_cond}
    \underset{u
    \in \bR^m}{sup}\left\{L_fh(x)+L_gh(x)u +\alpha (h(x)) \right\}\geq 0~~\forall x \in \bR^n
\end{align}
where $L_fh(x)$ and $L_gh(x)$ denote the Lie derivative of $h(x)$ in the direction of $f$ and $g$, respectively.
%\end{definition}
Given a ZCBF $h(x)$, define the set for all $x\in \Omega$
\begin{align}\label{Barrier_control}
    K_{zcbf}(x)&=\left\{u \in U:L_fh(x)+L_gh(x)u +\alpha (h(x))\geq 0\right\}.
\end{align}
%The following result holds.
\begin{theorem}[See \cite{XU201554, ames2019control}]\label{thm:ZCBF}
Consider system \eqref{mode:Nonlinear_system}, with $x\in \bR^n$ and $u \in U$
denoting the state and the input, respectively, f and g are
locally Lipschitz. Let $\Omega$ be a zero super level-set of a continuously differentiable function $h(x)$, defined in \eqref{Levelset:Omega}. If the relative-degree of \eqref{mode:Nonlinear_system} w.r.t $h(x)$ is $1$ (i.e., $L_gh(x)\neq 0$), then any Lipschitz continuous controller $u\in K_{zcbf}(x)$ renders the set $\Omega$ forward
invariant and asymptotically stable.
\end{theorem}
For many physical systems, including the one considered in the paper, the relative degree $r$ is greater than one.  In this case, as shown in \cite{xiao2019control}, one can recursively define the control barrier functions $h^i(x) = L_fh^{i-1}(x) +\alpha^I (h^{i-1}(x))$ and their zero super level sets $\Omega^i = \left\{x\in \bR^n|h^{i-1}(x)\geq 0\right\}$, where $i \in \left\{1 \ldots r-1\right\}$, $\alpha^1,\cdots,\alpha^{r-1}$ are all class $\mathcal{K}$ functions and $h^0(x) = h(x)$. This is the construction used in the paper.

% \begin{remark}\label{rem:cbf_relative_degree}
% For many physical systems including the one presented in the next section, the relative degree $r$ can be greater than $1$.  Consequently, as shown in \cite{xiao2019control}, one can recursively define $h^i(x) = L_fh^{i-1}(x) +\alpha^I (h^{i-1}(x))$ and $\Omega^i = \left\{x\in \bR^n|h^{i-1}(x)\geq 0\right\}$ as the new control barrier function and its zero super level set, respectively. Where $i \in \left\{1 \ldots r-1\right\}$, $\alpha^1,\cdots,\alpha^{r-1}$ are all class $\mathcal{K}$ functions and $h^0(x) = h(x)$.
% \end{remark}
%
%
%\end{appendix}
%\begin{appendix}
\section*{Appendix B: Proof of Theorem~\ref{prop:CBF_known_G}}
%\begin{proof}
Define the functions $b_l,~b_h:\bR^{2n} \rightarrow \bR^n$ as
\begin{align}\label{barrier_funs:b}
\begin{split}
    b_l &\triangleq V-v_l \mathbb{1}\\
    b_h &\triangleq v_h \mathbb{1} -V,
\end{split}
\end{align}
and their corresponding zero super level-set of \eqref{barrier_funs:b} as
\begin{align}\label{set:b}
    \Omega_b &\triangleq\left\{(I,V) \in \bR^{2n}|b_l \geq 0,~b_h\geq 0 \right\}.
\end{align}
Note that for all $(I,~V) \in \Omega_b$, the voltage $V$ satisfies Objective \ref{obj:Safe_voltage_regulation}. Since from Assumption~\ref{ass:feasibility}, the initial condition is in this set, our proof below will show that the controller in \eqref{QP:1} renders the set $\Omega_b$ forward invariant and asymptotically stable for the system \eqref{model:DC_microgid}.

To this end, consider $b_{l}$ as the candidate zeroing control barrier function. By choosing $C^{-1}G_p$ as the class $\mathcal{K}$ function, the condition~\eqref{Barrier_cond} for a zeroing control barrier function stated in the Appendix A simplifies to
\begin{equation}
\label{CBF:Known_G:a}
    \dot{b_l} + C^{-1}G_p b_l \geq 0.
\end{equation}
Multiplying $C$ on both the sides yields that~(\ref{CBF:Known_G:a}) is equivalent to the condition
\begin{equation}
    \label{CBF:Known_G:b}
    C\dot{V} + G_p (V-v_l \mathbb{1}) \geq 0.
\end{equation}
Using the relation $C\dot{V} + G_p V = I$ from~(\ref{model:DC_microgid}) and~(\ref{denote:Gp}), we can rewrite~(\ref{CBF:Known_G:b}) as
\begin{equation}
  \label{CBF:Known_G:c}
  I - v_lG_p \mathbb{1}   \geq 0.
\end{equation}
Finally, since $\mathcal{B}^\top 1 =0,$ we can simplify~(\ref{CBF:Known_G:c}) to the relation
\begin{equation}
\label{CBF:Known_G:d2_1}
    I - G v_l \mathbb{1}   \geq 0.
\end{equation}
Similarly, by considering $b_{h}$ as the zeroing control barrier function, we can obtain the condition
\begin{equation}
\label{CBF:Known_G:d2_2}
    - I + G v_h \mathbb{1}  \geq 0.
\end{equation}
Since the zero super level set of $b_{l}$ and $b_{h}$ is given by $\Omega_{b},$ Theorem~\ref{thm:ZCBF} implies that if~(\ref{CBF:Known_G:d2_1}) and~(\ref{CBF:Known_G:d2_2}) hold, then Objective~\ref{obj:Safe_voltage_regulation} is met.

However, there is no control input appearing in~(\ref{CBF:Known_G:d2_1}) and~(\ref{CBF:Known_G:d2_2}). This is a consequence of the fact that the relative degree of the system \eqref{model:DC_microgid} with respect to either  $b_l$ or $b_h$ is $2$. Following \cite{xiao2019control}, %As discussed in Remark \ref{rem:cbf_relative_degree} in the Appendix, 
we can overcome this issue by recursively defining zeroing control barrier functions. To this end, define functions $B_l,~B_h:\bR^{2n} \rightarrow \bR^n$, given by
\begin{align}
\begin{split}\label{CBF:Known_G:d}
B_l & \triangleq     I - v_lG  \mathbb{1} ,\\
B_h & \triangleq    - I + v_hG  \mathbb{1} ,
\end{split}
\end{align}
with the zero super level set 
\begin{align}\label{set:B}
    \Omega_B & \triangleq \left\{(I,V) \in \bR^{2n}|B_l \geq 0,~B_h\geq 0 \right\}.
\end{align}

To enforce~(\ref{CBF:Known_G:d2_1}), consider $B_{l}$ as the candidate zeroing control barrier function. By using $L^{-1}[\eta_l]$ as the class $\mathcal{K}$ function, we can write the condition in \eqref{Barrier_cond} as 
\begin{equation}
\label{CBF:Known_G:e}
    \dot{B}_l + L^{-1}[\eta_l] B_1 \geq 0.
\end{equation}
Multiplying both sides of~(\ref{CBF:Known_G:e}) with $L$ and simplifying by using~(\ref{model:DC_microgid}), we can rewrite this condition as
\begin{equation}
    \label{CBF:Known_G:g_1}
    - V + [V_s]u + [\eta_l](I - G v_l \mathbb{1})  \geq 0.
\end{equation}
Similarly, by considering $B_{h}$ as the candidate zeroing control barrier function, we can obtain the condition
\begin{equation}
    \label{CBF:Known_G:g_2}
    V - [V_s]u - [\eta_h] (I - G v_h \mathbb{1})   \geq 0.
\end{equation}

Thus, we have shown that enforcing~(\ref{CBF:Known_G:g_1}) and~(\ref{CBF:Known_G:g_2}) ensures that the set $\Omega_B$ is forward invariant, which further enforces that $\Omega_b$ is forward invariant and asymptotically stable and hence that Objective~\ref{obj:Safe_voltage_regulation} is met. But the first two constraints in the optimization problem \eqref{QP:1} are equivalent to ~(\ref{CBF:Known_G:g_1}) and~(\ref{CBF:Known_G:g_2}). Thus, we conclude that the controller $u_{i,1}^{\mathrm{opt}}$ designed in \eqref{QP:1} ensures that the system \eqref{model:DC_microgid} satisfies Objective \ref{obj:Safe_voltage_regulation}. \qedsymbol
% \end{proof}
% \end{appendix}

\section*{Appendix C: Proof of Theorem~\ref{prop:CBF_unknown_G}}
%\begin{proof}
Following the arguments from the proof of Theorem~\ref{prop:CBF_known_G}, we know that Objective~\ref{obj:Safe_voltage_regulation} is satisfied if $\Omega_b$ is rendered forward invariant. In turn, this can be ensured by a control input that yields $B_l,~B_h \geq 0$. Since the exact value of $G$ is unknown, we impose the sufficient condition that 
\begin{align}
    \begin{split}
        I &\geq v_{l} G\mathbb{1} \geq v_{l}G_{l}\mathbb{1},\\
        I &\leq v_{h} G\mathbb{1} \leq v_{h}G_{h}\mathbb{1}.
    \end{split}
\end{align}
In other words, define the two functions
\begin{align}\label{cbf:Bhat}
\begin{split}
\hat B_l &\triangleq     I - v_l G_l  \mathbb{1}   \geq 0,\\
\hat B_h &\triangleq   - I + v_h G_h  \mathbb{1}  \geq 0,
\end{split}
\end{align}
with the super level set
\begin{align}\label{set:B2}
    \hat \Omega_B & \triangleq \left\{(I,V) \in \bR^{2n}|\hat B_l \geq 0,~\hat B_h\geq 0 \right\},
\end{align}
which satisfies $\hat{\Omega}_B\subset \Omega_B$. We now consider $\hat{B}_l$ and $\hat{B}_h$ as the zeroing control barrier functions. Following the proof of Theorem~\ref{prop:CBF_unknown_G}, we can derive the conditions
\begin{align}
\begin{split}\label{CBF:unKnown_G:a}
    - V + [V_s]u + [\eta_l](I - G_l v_l \mathbb{1})  &\geq 0,\\
     V - [V_s]u - [\eta_h] (I - G_h v_h \mathbb{1})   &\geq 0,
\end{split}
\end{align}
in place of~(\ref{CBF:Known_G:g_1}) and~(\ref{CBF:Known_G:g_2}. The conditions~\ref{CBF:unKnown_G:a} are equivalent to the  first two constraints in the optimization problem \eqref{QP:2}. Thus, the controller $u_{i,2}^{\mathrm{opt}}$ designed in \eqref{QP:2} renders the set $\hat\Omega_B \subset \Omega_B$, and in turn, $\Omega_b$ forward invariant and asymptotically stable. Thus, the controller $u_{i,2}^{\mathrm{opt}}$ ensures that the system \eqref{model:DC_microgid} satisfies Objective \ref{obj:Safe_voltage_regulation} in spite of lack of precise knowledge of the load $G$. \qedsymbol
%\end{proof}
\end{document}